\documentclass{article}

\usepackage[preprint]{neurips_2024}

\usepackage[utf8]{inputenc} 
\usepackage[T1]{fontenc}
\usepackage{microtype}      

\newcommand{\varsub}[1]{_{\smash{#1}}} 
\newcommand{\us}[1]{^{\smash{#1}}}   
\newcommand{\ustime}{\us{(t)}} 

\newcommand{\capL}{\mathcal{L}}

\newcommand{\E}{\mathbb{E}}
\newcommand{\R}{\mathbb{R}}
\newcommand{\N}{\mathbb{N}}

\newcommand{\intN}{N}

\newcommand{\intT}{T}

\newcommand{\pearldo}{\texttt{do}}

\newcommand{\vw}{\mathbf{w}}
\newcommand{\vx}{\mathbf{x}}

\usepackage{pgfplots}        
\pgfplotsset{compat=newest}
\usepackage{tikz}            
\usetikzlibrary{
    calc,
    positioning,
    fit,
    tikzmark,
    arrows.meta,
    shapes,
    decorations.pathreplacing,
    intersections,
    through,
    matrix
}
\tikzset{
    line/.style = {
        thick,
        ->,
        > = {
            Triangle[length=1.5mm, width=1.5mm]
        }
    },
    arrow/.style = {
        line
    },
    hidden node/.style = {
        circle,
        minimum size = 1cm,
        inner sep=0pt,
        text width=10mm,
        draw = none,
        align = center,
        thick
    },
    variable node/.style = {
        hidden node,
        draw = black,
        fill = black,
        text = white,
        align = center,
    },
    latent node/.style = {
        hidden node,
        draw = black,
        fill = white,
        text = black,
        align = center,
    },
    factor node/.style = {
        hidden node,
        rectangle,
        draw = black,
        align = center,
    },
    observed node/.style = {
        variable node,
        fill = gray!30,
        text = black,
        align = center,
    },
}

\renewcommand{\paragraph}[1]{\par\vspace{1em}\noindent\textbf{\textit{#1.}}}
\newcommand{\directionref}[2]{\hyperref[#1]{\textbf{[RD#2]}}}
\newcommand{\challengeref}[2]{\hyperref[#1]{\textbf{[C#2]}}}

\newcommand{\ubar}[1]{\underaccent{\bar}{#1}}

\usepackage[utf8]{inputenc} 
\usepackage[T1]{fontenc}    
\usepackage{hyperref}       
\usepackage{url}            
\usepackage{booktabs}       
\usepackage{amsfonts}       
\usepackage{nicefrac}       
\usepackage{microtype}      
\usepackage{xcolor}         

\usepackage{mathtools}
\usepackage{subcaption}
\usepackage{amssymb}
\usepackage{mathtools}

\usepackage{tikz}
\tikzset{
mystyle/.style={
  circle,
  inner sep=0pt,
  text width=10mm,
  align=center,
  draw=black,
  fill=white,
  line width=0.3mm
  }
}

\usepackage{accents}

\usepackage{amsthm}
\theoremstyle{plain}
\newtheorem{theorem}{Theorem}[section]
\newtheorem{proposition}[theorem]{Proposition}

\theoremstyle{definition}
\newtheorem{definition}[theorem]{Definition}

\title{Towards Replication-Robust Analytics Markets}

\author{%
  Thomas Falconer\thanks{Correspondence to: \texttt{falco@dtu.dk}} \\
  Technical University of Denmark, DK \\
  \And
  Jalal Kazempour \\
  Technical University of Denmark, DK\\
  \AND
  Pierre Pinson \\
  Imperial College London, UK\\
}

\hypersetup{
    colorlinks=true,
    linkcolor=blue,
    urlcolor=orange,
    citecolor=magenta,
    pdftitle={Towards Replication-Robust Analytics Markets},
    pdfpagemode=FullScreen,
}

\begin{document}

\maketitle

\begin{abstract}
    \textit{Despite recent advancements in machine learning, in practice, relevant datasets are often distributed among market competitors who are reluctant to share. To incentivize data sharing, recent works propose analytics markets, where multiple agents share features and are rewarded for improving the predictions of others. These rewards can be computed by treating features as players in a coalitional game, with solution concepts that yield desirable market properties. However, this setup incites agents to strategically replicate their data and act under multiple false identities to increase their own revenue and diminish that of others, limiting the viability of such markets in practice. In this work, we develop an analytics market robust to such strategic replication for supervised learning problems. We adopt Pearl's do-calculus from causal inference to refine the coalitional game by differentiating between observational and interventional conditional probabilities. As a result, we derive rewards that are replication-robust by design.}
\end{abstract}

\section{Introduction} \label{sec:intro}
Machine learning relies heavily on the quality and quantity of input data, however firms often find it difficult, if not impossible, to acquire rich datasets themselves. This is often due to privacy constraints. For instance, in the medical domain, data is highly sensitive and subject to strict regulations \citep{rieke2020future}, yet hospitals could benefit from sharing patient information to mitigate social biases in diagnostic support systems. Similar examples include rival distributors sharing sales data to improve supply forecasts, or hotel operators using airline data to better anticipate demand.
One promising solution to this problem is \textit{federated learning}, 
where a central model (e.g., a neural network) is trained by multiple distributed agents without centralizing any data \citep{zhang2021survey}. Instead, only model parameters (e.g., weights and biases) are shared, with the option to include differential privacy by design \citep{wei2020federated}. 

However, this still assumes that data owners will collaborate altruistically---an assumption that may not hold if these agents also compete in downstream markets \citep{gal1985information}.
To incentivize data sharing, one can instead frame data as a commodity within a market-based framework \citep{bergemann2019markets}. Whilst many platforms already exist to purchase raw datasets directly from their owners via bilateral transactions \citep{rasouli2021data}, pricing these datasets is not easy as their value ultimately depends on when, how, and by whom they are eventually used. Further, since datasets often contain overlapping information, their value is inherently combinatorial.
With this in mind, recent works instead advocate for \emph{analytics markets}---real-time mechanisms that match datasets to machine learning tasks based on predictive performance \citep{agarwal2019marketplace}.
In these markets, a central platform collects features from multiple sellers, and buyers post machine learning tasks along with bids that reflect their willingness
to pay for marginal improvements in accuracy. The platform processes these inputs and determines what information the buyer should receive, and what price they should pay. These payments establish the market revenue, which is allocated amongst the sellers, rewarding them in proportion to their contributions to the improved accuracy. Importantly, buyers only receive refined predictive models, rather than raw features, positioning
these markets as mechanisms to incentivize federated learning.
Such markets have been proposed for both classification \citep{koutsopoulos2015auctioning} and regression \citep{pinson2022regression} tasks.

To allocate revenue, each feature can be treated as a player in a coalitional game, for which well-established solution concepts can be applied, namely semivalues \citep{dubey1981value}, which are characterized by a set of axioms---symmetry, efficiency, null-player, and additivity---that lead to desirable market properties by design (see \citet{chalkiadakis2011computational} for precise definitions of these properties). A feature’s semivalue represents it's expected marginal contribution to predictive performance given all subsets (or coalitions) of other features. The Shapley value \citep{shapley1997value} is a particularly appealing semivalue as it is the only one which satisfies all four axioms.

\subsection{Challenges}
For any feature vector $\vx \in \mathbb{R}\us{\intN}$, a revenue allocation policy should ideally be $\boldsymbol{\phi} : \mathcal{L} \times \mathbb{R}\us{\intN} \mapsto \mathbb{R}\us{\intN}$, where $\mathcal{L}$ is a set of possible scoring rules $\ell : \mathbb{R}\us{\intN} \mapsto \mathbb{R}$ that, given an observation, map the feature vector to a real value. In other words, the output of $\ell(\vx)$ is decomposed into contributions $\boldsymbol{\phi} (\ell, \vx) = \left(\phi\varsub{1}, \dots, \phi\varsub{\intN} \right)$ for each feature, such that $\ell$ need only be evaluated for the compete vector $\vx$. However, to compute the Shapley values, the scoring rule needs to be evaluated for $2\us{\intN}$ coalitions of features, yet many machine learning models cannot easily produce outputs for partial inputs due to matrix dimension mismatches.

To address this, one must also define a so-called \textit{lifting} function $\xi : \mathcal{L} \times \mathbb{R}\us{\intN} \times 2\us{\intN} \mapsto \mathbb{R}$, which extends $\ell$ to operate on subsets $\omega \subseteq \{1, \dots, \intN\}$ of features \citep{merrill2019generalized}. That is, $\xi(\ell, \vx, \omega)$ assigns a value for each $\omega$, so lifts the scoring rule $\ell$ from the original domain $\mathbb{R}\us{\intN}$ to $\mathbb{R}\us{\intN} \times 2\us{\intN}$, which simulates the removal of features to compute partial score evaluations by averaging over the out-of-coalition features according to some probability distribution.
However, there are many distributions to choose from to achieve this, leaving open the question of which is most appropriate for revenue allocation in analytics markets. In this paper, we approach this choice through the lens of causality, specifically, we consider whether the distribution over out-of-coalition features should be based on \textit{observational} or \textit{interventional} conditional probabilities. From a causal inference perspective, observational conditioning reflects the expected model output given the observed values of the in-coalition features, while interventional conditioning reflects the expected model output when one actively “intervenes" to set those features to specific values.

In existing works, the choice of distribution is observational (e.g., \citealp{agarwal2019marketplace}). Interestingly, these works also reveal a vulnerability to malicious behavior where agents can submit replicates of their features under different identities in the hope that they are highly correlated with features from other agents, as in this case they can increase their revenue and diminish that of others. 
For instance, consider a market with two sellers, each selling one feature which are both identical. One would naturally expect for any revenue to be split equally. However, if one seller replicated their feature once and sold it again in the market under a false identity, they would now receive two thirds of the revenue whilst the other only one third, without providing any additional predictive performance, thus the market is not robust to replication.

Various attempts have been made to remedy this. For example, \citet{ohrimenko2019collaborative}’s more elaborate mechanism is robust to replication but requires each seller to have their own analytics task (or prediction problem), posing practical challenges. \citet{agarwal2019marketplace} modify the Shapley value to penalize similar features, however this comes at the cost of budget balance, with some revenue remaining unallocated. Their proposal also remains vulnerable to spiteful agents who are willing to reduce others’ revenue at the expense of their own. The key contribution of our work is a new market design that replaces observational conditioning with interventional conditioning, and we prove that this approach guarentees replication-robust rewards by design.

\subsection{Contributions}
The key contributions of our work are as follows: (i) we propose a general analytics market design for supervised learning problems that subsumes recent proposals in literature; (ii) we show that there are many ways in which Shapley values can be used to allocate revenue and that the differences between them can be explained from a causal perspective; (iii) we show that the replication incentives in existing works can be explained using \citet{pearl2012calculus}'s seminal work on causality; (iv) by replacing the conventional approach of conditioning by \textit{observation} with conditioning by \textit{intervention}, we design a market that is robust to replication whilst also accounting for spiteful agents, thereby taking a step toward the practical application of these markets; and finally (v) we demonstrate our findings on a real-world case study---out of many potential applications, we choose to study wind power forecasting due to data availability, the known value of sharing distributed data, and the fact it is a sandbox that can be easily shared and used by others.

The remainder of this paper is structured as follows: Section~\ref{sec:preliminaries} presents our general market framework. 
In Section~\ref{sec:lift_formulations} we derive variants of the characteristic function and analyze each from a causal perspective.
In Section~\ref{sec:robust} we discuss the impact of each on the robustness of the market to replication.
Section~\ref{sec:experimental_analysis} then illustrates our findings on a real-world case study.
Finally, Section~\ref{sec:conclusions} gathers a set of conclusions and perspectives for future work.

\section{Market Framework} \label{sec:preliminaries}
We focus on analytics markets in which the buyers’ post regression tasks, namely regression markets, however our framework can be adapted to any supervised learning problem. This builds on the seminal work of \citet{dekel2010regression}, who were one of the first to study data acquisition mechanisms for regression tasks where agents might behave strategically when sharing private data.
We model a single buyer and multiple sellers, which naturally extends to parallel, independent market instantiations for multiple buyers. Given a finite set of agents, we label the one acting as the buyer at any instantiation as the central agent and the remaining agents as support agents. The central agent’s valuation for predictive accuracy reflects, for instance, their perceived cost of forecast errors in a downstream decision-making process.
We denote this valuation $\lambda \in \R\varsub{\geq 0}$, the value of which we assume to be known and reported truthfully. We refer the reader to \citet{ravindranath2024data} for a recent proposal of how $\lambda$ may be elicited in practice.

\paragraph{Central Agent}
The central agent targets a stochastic process $\{Y\ustime\}$, from which they observe a time series $\{y\ustime\}$, with each $y\ustime \in \R$ a realization from $Y\ustime$ at time $t \in \N$. Instead of targeting a specific functional of $Y\ustime$, such as the expected value or a particular quantile, we model the entire distribution, conditioned on the available features. Any particular summary statistic extracted by the central agent is treated as part of the downstream decision-making process.
The central agent owns a vector of $M$ features, the values of which at time $t$ are denoted by $\vx\ustime = (x\ustime\varsub{1}, \dots, x\ustime\varsub{M})\us{\top}$.

\paragraph{Regression Task}
The central agent posts a \textit{regression task} to the market platform. This includes their observed target and own feature variables, as well as a latent variable model of the data-generating process, which specifies a prior $p(\vw)$ over latent variable $\vw$ and a likelihood $p(y\ustime \vert \vx\ustime, \vw)$, with the likelihood assumed to be a Gaussian distribution with expected value:
\begin{equation}
   f(\vx\us{(t)}, \boldsymbol{w}) = \vw\us{\top} \vx\ustime, \label{eq:central_agent_likelihood}
\end{equation}
with the variance treated as a hyperparameter. Specifically, we focus on parametric regression with functions that are linear in their coefficients to guarantee certain market properties. One can of course obtain a rich class of models with linear combinations of nonlinear basis functions or splines, however we adopt only a linear basis in this work (see \citet{falconer2024bayesian} for an application of nonlinear basis functions to analytics markets).

The central agent can infer posterior beliefs with their own features---after observing ($\vx\ustime, y\ustime$), they can update prior beliefs over the latent variable $\vw$ via Bayes’ rule, with the posterior given by
\begin{equation}
  p(\vw \vert y\ustime, \vx\ustime) =
  \frac{p(y\ustime \vert \vx\ustime, \vw) p(\vw)}{\int p(y\ustime, \vw \vert \vx\ustime) d\vw}.
  \label{eq:central_agent_posterior}
\end{equation}
In time-series regression analysis, observations often arrive sequentially, so the posterior at time $t-1$ becomes the prior at time $t$, such that $ p(\vw \vert  y\ustime, \vx\ustime) \propto p(y\ustime \vert \vx\ustime, \vw) p(\vw \vert y\us{(t-1)}, \vx\us{(t-1)})$. Greater weight can be placed on more recent data by augmenting this update step to use exponential forgetting, recasting it as a trade-off between the posterior from the previous time step and the original prior, modeling a gradual erosion of confidence in past data. Specifically, the prior at time $t+1$ is replaced by $ p(\vw \vert y\ustime, \vx\ustime; \tau) \propto p(\vw \vert y\ustime, \vx\ustime)\us{\tau} p(\vw)\us{1-\tau}$,
where $\tau \in \R\varsub{(0, 1)}$ is analogous to the forgetting factor in the special case of time-weighted Least squares, which uses exponential decay to assign greater weight to more recent time indices.

\clearpage
We assume a centered isotropic Gaussian prior, meaning the resulting posterior is also Gaussian due to conjugacy with the likelihood in \eqref{eq:central_agent_likelihood}. This Bayesian approach to regression subsumes many frequentist methods, making it easy to apply, for instance, ordinary least-squares, for which a Gaussian likelihood is an implicit assumption.

Before updating beliefs, the current posterior is used for out-of-sample forecasting. Specifically, to make a prediction for time $t+1$, the predictive distribution is obtained by marginalizing out $\vw$ with respect to the posterior in \eqref{eq:central_agent_posterior}, such that
\begin{equation}
     \hat{y}\us{(t+1)} = \int p(y\us{(t+1)} \vert \vx\us{(t+1)}, \vw) p(\vw \vert y\ustime, \vx\ustime) d\vw,
     \label{eq:bayesian_prediction}
\end{equation}
where $\hat{y}\us{(t+1)} = p(y\us{(t+1)} \vert \vx\us{(t+1)})$, which too is Gaussian, with variance equal to the sum of the variance of the noise and the posterior uncertainty.

The final part of the regression task is a scoring rule $\ell \in \mathcal{L}$ used to evaluate performance, where $\mathcal{L}$ is the family of negatively oriented, strictly proper scoring rules.
In an online setup, evaluating $\ell$ at each time step can be viewed as a recursive and adaptive estimation of its expected value, in the sense that a greater weight is placed on more recent data. The estimate of $\E[\ell]$ at time $t$ using the central agent's own features can be described by the following recursion:
\begin{equation}
    \E[\ell]\ustime = (1 - \tau) \ell\ustime + \tau \E[\ell]\us{(t-1)} \label{eq:recursive_loss}
\end{equation}

\paragraph{Support Agents}
Suppose there are $D$ additional features available in the market distributed amongst the support agents, such that the full feature vector at time $t$ is given by 
\begin{equation*}
    \vx\ustime = (x\ustime\varsub{1}, \dots, x\ustime\varsub{M}, x\ustime\varsub{M+1}, \dots, x\ustime\varsub{M+D})\us{\top},
\end{equation*}
where the first $M$ entries are the central agent’s own features, and the remaining $D$ entries belong to support agents, indexed by the ordered set $\Omega = \{M+1, \dots, M+D\}$.

Each of the support agents $a \in \mathcal{A}\varsub{-c}$ owns a subset $\Omega_a \subseteq \Omega$ of indices, such that $D = \sum\varsub{a \in \mathcal{A}\varsub{-c}} \vert\Omega_a\vert$. 
To lighten notation, for any subset of indices $\omega \subseteq \Omega$, we write $\vx\ustime\varsub{\omega}$ as the vector at time $t$ consisting of the central agent's features together with the features indexed by $\omega$, such that
\begin{equation*}
    \vx\ustime\varsub{\omega} = (x\ustime\varsub{i})\varsub{\forall i \in [M]}\us{\top} \oplus (x\ustime\varsub{j})\varsub{\forall j \in \omega}\us{\top}.
\end{equation*}

The central agent’s latent variable model is then extended to incorporate the additional features, such that the expected value of the likelihood becomes:
\begin{align*}
    f(\vx\us{(t)}, \boldsymbol{w}) = \underbracket{w_0 + \sum\varsub{i \in [M]} w\varsub{i} x\varsub{i}\us{(t)}}_{\substack{\textnormal{Terms belonging}\\\textnormal{to the central agent.}}} + \underbracket{\sum\varsub{a \in \mathcal{A}\varsub{-c}} \sum\varsub{j \in \Omega_a} w\varsub{j} x\varsub{j}\us{(t)}}_{\substack{\textnormal{Terms belonging}\\\textnormal{to the support agents.}}}.
    \label{eq:interpolated_function}
\end{align*}

For any $\omega \subset \Omega$, we write $\hat{y}\ustime\varsub{\omega}$ for the prediction at time $t$ based on features $\vx\ustime\varsub{\omega}$, with $\ell\ustime\varsub{\omega}$ the associated score, such that $\ell\ustime\varsub{\Omega}$ measures predictive performance using all available features.

\paragraph{Market Clearing} 
Once a regression task is posted by the central agent, the market is ready to be cleared, which, as in real-world MLOps pipelines, is separated into training and testing stages. In the training stage, Bayesian inference is applied to observed data; in the test stage, the trained model is used for forecasting on previously unseen data. Any feature's value is thus determined by its marginal contributions to both the in-sample \textit{and} out-of-sample score. 

The market revenue is a function of the exogenous valuation, $\lambda$, and the extent to which model-fitting is improved. This is measured using the current estimate of expected value of the scoring rule, such that any time $t$, the market revenue is
\begin{equation*}
    r\ustime\varsub{c} = \lambda \left( \E[\ell\varsub{\varnothing}]\ustime - \E [\ell\varsub{\Omega}]\ustime \right), \label{eq:market_revenue}
\end{equation*}
where $\ell\varsub{\varnothing}$ is the evaluation of the scoring rule using only the central agent's features, and $r\ustime\varsub{c}$ is also the payment collected from the central agent. Once this has been collected, the next step is to reward the support agents. To do this, each feature is first portioned a share of the market revenue using an allocation policy $\boldsymbol{\phi}$, which computes the marginal contribution of each to the gain in predictive performance. 
The sum of these allocations for the features belonging to a seller is their reward, denoted by
\begin{equation}
    r\ustime\varsub{a} = r\ustime\varsub{c} \sum\varsub{i \in \Omega_a}  \E[\phi_i]\ustime, \label{eq:reward_rule}
\end{equation}
where $\boldsymbol{\phi} \in \Phi$ is a probability simplex such that 
\begin{equation*}
    \Phi = \left\{ \boldsymbol{\phi} \in \mathbb{R}\us{D} : \phi_i \geq 0, \sum\varsub{i=1}\us{D} \phi\varsub{i} = 1  \right\},
\end{equation*}
where $\phi\varsub{i}$ is the gain in performance contributed by $x\ustime\varsub{i}$ for every $i \in \Omega$, as this would ensure no support agent loses money and all revenue is allocated.

To formulate $\boldsymbol{\phi}$, we turn to coalitional game theory, which provides a principled framework to divide utility amongst cooperating players. Treating the $D$ features owned by support agents as players in a coalitional game, we define a set function $\xi : \capL \times \R\us{M + D} \times 2^{\Omega} \to \R$, where $\xi(\ell, \vx\varsub{\Omega}\ustime, \omega)$ is the score achieved by each coalition $\omega \subseteq \Omega$. Hence, $\xi$ may also be referred to as a lifting function which lifts the scoring rule $\ell$ from its original domain $\R\us{M + D}$ to $\R\us{M + D} \times 2\us{\Omega}$.
This requires simulating the removal of features to enable partial evaluations of $\ell$, which is not straightforward. As a result, formulating $\xi$ is also challenging. We defer a detailed exploration of its formulation and the associated difficulties to Section~\ref{sec:lift_formulations}.
For a given $\xi$, the the Shapley value can be written as
\begin{equation} 
    \phi\ustime\varsub{i} = \frac{1}{D}\sum\varsub{\omega \subseteq \Omega \setminus i}  \binom{D-1}{\vert \omega \vert}\us{-1} \delta\ustime\varsub{i} (\omega),
    \label{eq:shapley_ml}
\end{equation}
where $\delta\varsub{i}\ustime (\omega) = \xi\ustime\varsub{\omega} - \xi\ustime\varsub{\omega \cup i}$ is the marginal contribution of feature $i$ to coalition $\omega$, having written $\xi\ustime\varsub{\omega} = \xi(\ell, \vx\ustime\varsub{\Omega}, \omega)$ for brevity.
Evaluating \eqref{eq:shapley_ml} exactly requires summing over all $2\us{D}$ feature subsets, resulting in exponential time complexity, and is known to be NP-hard \citep{deng1994complexity}.
Hence, in practice, one must generally rely on approximation methods \citep{castro2009polynomial, mitchell2022sampling}. 
However, in our work we are primarily focused on the functional form of $\xi$, which is agnostic to the choice of sampling method, so exploring state-of-the-art approximations is out of scope. 

\begin{definition}[\textbf{Market properties}] \label{def:properties}
    With the proposed regression framework and Shapley value-based revenue allocation, regression markets have the following properties:
    \begin{enumerate}
    \item \textit{Symmetry}---Any features that have equal contribution to all coalitions obtain equal reward, i.e., $\forall \omega \in \Omega\setminus \{i, j\} : \xi\ustime\varsub{\omega \cup i} \equiv \xi\ustime\varsub{\omega \cup j} \mapsto \phi_{i}\ustime \equiv \phi_{j}\ustime, \forall (i, j) \in \Omega, i \neq j$, which means allocations are invariant to permutation of indices.
    \item \textit{Linearity}---For any two features, their joint contribution to coalition is equal to the sum of their marginal contributions, i.e., $\xi\ustime\varsub{\omega \cup i} + \xi\ustime\varsub{\omega \cup j} = \xi\ustime\varsub{\omega \cup i, j}, \forall (i, j) \in \Omega$, ensuring that rewards are consistent if features are offered individually or as a bundle, removing any incentive to strategically package features.
    \item \textit{Budget balance}---The payment of the central agent is equal to the total sum of rewards received by all the support agents, i.e., $r\varsub{c}\ustime = \sum_{a \in \mathcal{A}_{-c}} r\varsub{a}\ustime$, which ensures all market revenue is allocated.
    \item \textit{Individual rationality}---Support agents have a weak preference to participate in the market rather than the outside option, i.e., $r\varsub{a}\ustime \geq 0, \forall a \in \mathcal{A}_{-c}$, meaning that no agent loses money.
    \item \textit{Zero-element}---If a support agent provides no feature, or provide features with zero marginal contribution to all coalitions, they earn no reward, i.e., $\forall \omega \in \Omega: \xi\ustime \varsub{\omega \cup i} \equiv \xi\ustime\varsub{\omega}, \forall i \in \Omega\varsub{a} \mapsto r_{a} = 0$.
    \item \textit{Truthfulness}---Support agents maximize their reward by reporting their true data.
\end{enumerate}
\end{definition}

\clearpage
These desirable market properties stem from the axioms of the Shapley value, a detailed proof of which is provided in \citet{falconer2024bayesian}. 

\section{Lifting Function} \label{sec:lift_formulations}
To evaluate $\ell$ for each subset (or coalition) of features $\omega \subseteq \Omega$, the lift averages the values of the remaining (out-of-coalition) features with respect to some probability distribution. It is not immediately clear which distribution \textit{should} be used as there are many options to choose from \citep{sundararajan2020many}, however most can be categorized as either \textit{observational} or \textit{interventional}. 

An observational lift uses the \textit{observational conditional expectation}, the expected score at time $t$, where the integral is taken with respect to the out-of-coalition features, conditional on in-coalition features taking on their observed values, such that
\begin{equation}
    \begin{aligned}
        \xi\us{(t), \textrm{obs}}\varsub{\omega} = 
        \int \ell(\vx\ustime\varsub{\omega}, \vx\ustime\varsub{\overline{\omega}}) p(\vx\ustime\varsub{\overline{\omega}}\vert \vx\ustime\varsub{\omega}) 
        d\vx\ustime\varsub{\overline{\omega}},
        \label{eq:observational_shapley}
    \end{aligned}
\end{equation}
where $\overline{\omega} = \Omega \setminus \omega$ denotes the out-of-coalition features.

The interventional lift instead uses the \textit{interventional conditional expectation}, given by
\begin{equation}
    \begin{aligned}
       \xi\us{(t), \textrm{int}}\varsub{\omega}
        &= 
        \int \ell(\vx\ustime\varsub{\omega}, \vx\ustime\varsub{\overline{\omega}}) p(\vx\ustime\varsub{\overline{\omega}}\vert \pearldo(\vx\ustime\varsub{\omega})) 
        d\vx\ustime\varsub{\overline{\omega}}, 
        \label{eq:interventional_shapley}
    \end{aligned}
\end{equation}
where $\pearldo(\cdot)$ is an operator from \citet{pearl2012calculus}'s \textit{do}-calculus.
In causal reasoning theory, the observational conditional probability in \eqref{eq:observational_shapley} describes the relationship between variables as they occur naturally, whereas the interventional conditional probability in \eqref{eq:interventional_shapley} is the result of ``intervening" by fixing a particular variable’s value \citep{pearl2010introduction}.
The key difference between \eqref{eq:observational_shapley} and \eqref{eq:interventional_shapley} is that in the former, conditioning on the observed values of the features in the coalition can alter the distribution of the out-of-coalition features if any indirect dependencies exist, whereas in the latter, the distribution of the out-of-coalition features is unaffected by the \textit{do}-intervention.

\begin{figure}[!t]
\centering
\begin{tikzpicture}[]
    \node (Z) at (0, 2) [latent node] {$Z\ustime$};
    \node (X1) at (-1, 0) [variable node] {$X\ustime\varsub{1}$};
    \node (X2) at (1, 0) [variable node] {$X\ustime\varsub{2}$};
    \draw[arrow] (Z) -- (X1) {};
    \draw[arrow] (Z) -- (X2) {};
    \draw[arrow, dashed] (X1) to[out=30,in=150] (X2) {};
    \draw[arrow, dashed] (X2) to[out=210,in=330] (X1) {};
\end{tikzpicture}
\caption{Direct (solid) and indirect (dashed) causal effects.Observed and latent variables are colored in black and white, respectively.}
\label{fig:backdoor_example}
\end{figure}

For further intuition, consider the graphical model in Figure~\ref{fig:backdoor_example}. The latent variable $Z\ustime$ is a confounder that causes both $X\ustime\varsub{1}$ and $X\ustime\varsub{2}$. Although there is no direct causal path from $X\ustime\varsub{1}$ to $X\ustime\varsub{2}$, there is an indirect backdoor path: $X\ustime\varsub{2} \leftarrow Z\ustime \to X\ustime\varsub{1}$,
which induces a statistical correlation between the two variables.
If we observe $X\ustime\varsub{1} = x\ustime\varsub{1}$, the resulting observational conditional distribution over $X\ustime\varsub{2}$, denoted $p(x\ustime\varsub{2} \vert x\ustime\varsub{1})$, reflects the correlation induced by the shared confounder $Z\ustime$. 
The interventional conditional distribution $p(x\ustime\varsub{2} \vert \pearldo(x\ustime\varsub{1}))$ instead describes what would happen if we artificially set $X\ustime\varsub{1}$ to $x\ustime\varsub{1}$, removing all incoming edges into $X\ustime\varsub{1}$, thereby blocking the backdoor path through $Z\ustime$. This isolates the direct casusal effect of $X\ustime\varsub{1}$ on $X\ustime\varsub{2}$, which is null here, so $p(x\ustime\varsub{2} \vert \pearldo(x\ustime\varsub{1})) = p(x\ustime\varsub{2})$. 

\paragraph{Marginal Expectations}
Whilst in this example, the interventional distribution $p(x\ustime\varsub{2} \vert \pearldo(x\ustime\varsub{1}))$ coincides with the marginal distribution $p(x\ustime\varsub{2})$, this equivalence is specific to the causal structure of the example and does not hold in general.
However, \citet{janzing2020feature} showed that to compute the marginal contributions of features in a machine learning model, it is natural to consider the inputs of the model as \textit{causes} of the output. Whilst this a seemingly trivial remark, the authors emphasize that to analyze what happens when the model inputs are changed, rather than the true features, the causal relations of concern are not those that appear between any features in the real world, but only those in the input-output system described by the machine learning model at hand.  
For instance, consider the \textit{true} data-generating process in Figure~\ref{fig:janzing_process}, where $X\ustime\varsub{1}$ directly causes $Y\ustime$, but $X\ustime\varsub{2}$ affects it only indirectly via latent confounder $Z\ustime$. The system output is the score $\ell\ustime$ via prediction $\hat{Y}\ustime$ and the \textit{true} features $X\ustime\varsub{i}$ are latent variables (as shown in Figure~\ref{fig:graphical_model}) that cause the model inputs $\tilde{X}\ustime\varsub{i}$ which are plugged into the regression model.

Observing $\tilde{X}\ustime\varsub{1} = \tilde{x}\ustime\varsub{1}$ (as in Figure~\ref{fig:janzing_observation}) preserves the backdoor path, so to compute the expected score conditional on this observation, one needs to integrate over $\tilde{X}\ustime\varsub{2}$ with respect to $p(\tilde{x}\ustime\varsub{2} \vert \tilde{x}\ustime\varsub{1})$, capturing the correlations induced by the shared confounder.
In contrast, intervening to fix $\tilde{X}\ustime\varsub{1} = \tilde{x}\ustime\varsub{1}$ (as in Figure~\ref{fig:janzing_intervention}) severs the edge from $Z\ustime$, isolating the causal effect of $\tilde{X}\ustime\varsub{1}$ on $\ell\ustime$. The backdoor is blocked, and the expectation is taken over the marginal, and thus $p(\tilde{x}\ustime\varsub{2} \vert \pearldo(\tilde{x}\ustime\varsub{1})) = p(\tilde{x}\ustime\varsub{2})$. 
Hence, in a machine learning context, interventional expectations coincide with marginal expectations in general.
As a result, when features are independent, the two lifts are equivalent, since $p(\tilde{x}\ustime\varsub{2} \vert \tilde{x}\ustime\varsub{1}) = p(\tilde{x}\ustime\varsub{2})$ in this case.

\begin{figure}
\centering
    \begin{subfigure}{0.24\textwidth}
    \centering
    \begin{tikzpicture}
        \node (target) at (0, 0) [variable node] {$Y\ustime$};
        \node (X1) at (-1, 2) [variable node] {$X\ustime\varsub{1}$};
        \node (X2) at (1, 2) [variable node] {$X\ustime\varsub{2}$};
        \node (Z0) at (0, 4) [latent node] {$Z\ustime$};
        \draw[arrow] (X1) -- (target) {};
        \draw[arrow, dashed] (X2) -- (target) {};
        \draw[arrow] (Z0) -- (X1) {};
        \draw[arrow] (Z0) -- (X2) {};
        \end{tikzpicture}
        \caption{Data}
        \label{fig:janzing_process}
        \end{subfigure}
        \hfill
        \begin{subfigure}{0.24\textwidth}
        \centering
        \begin{tikzpicture}
        \node (pred) at (-1, 0) [variable node] {$\hat{Y}\ustime$};
        \node (loss) at (1, 0) [variable node] {$\ell\ustime$};
        \node (X1) at (-1, 2) [variable node] {$\tilde{X}\ustime\varsub{1}$};
        \node (X2) at (1, 2) [variable node] {$\tilde{X}\ustime\varsub{2}$};
        \node (Z1) at (-1, 4) [latent node] {$X\ustime\varsub{1}$};
        \node (Z2) at (1, 4) [latent node] {$X\ustime\varsub{2}$};
        \node (Z0) at (0, 6) [latent node] {$Z\ustime$};
        \draw[arrow] (Z0) -- (Z1) {};
        \draw[arrow] (Z0) -- (Z2) {};
        \draw[arrow] (Z1) -- (X1) {};
        \draw[arrow] (Z2) -- (X2) {};
        \draw[arrow] (X1) -- (pred) node[black, midway, xshift=-5mm, yshift=0mm] {$w_1$};
        \draw[arrow] (X2) -- (pred) node[black, midway, xshift=6mm, yshift=0mm] {$w_2$};
        \draw[arrow] (pred) -- (loss) {};
    \end{tikzpicture}
    \caption{Model}
    \label{fig:graphical_model}
    \end{subfigure}
    \hfill
    \vspace{5mm}
    \begin{subfigure}{0.24\textwidth}
    \centering
    \begin{tikzpicture}
        \node (pred) at (-1, 0) [variable node] {$\hat{Y}\ustime$};
        \node (loss) at (1, 0) [variable node] {$\ell\ustime$};
        \node (X1) at (-1, 2) [observed node] {$\tilde{x}\ustime\varsub{1}$};
        \node (X2) at (1, 2) [variable node] {$\tilde{X}\ustime\varsub{2}$};
        \node (Z1) at (-1, 4) [latent node] {$X\ustime\varsub{1}$};
        \node (Z2) at (1, 4) [latent node] {$X\ustime\varsub{2}$};
        \node (Z0) at (0, 6) [latent node] {$Z\ustime$};
        \draw[arrow] (Z0) -- (Z1) {};
        \draw[arrow] (Z0) -- (Z2) {};
        \draw[arrow] (Z1) -- (X1) {};
        \draw[arrow] (Z2) -- (X2) {};
        \draw[arrow] (X1) -- (pred) node[black, midway, xshift=-5mm, yshift=0mm] {$w_1$};
        \draw[arrow] (X2) -- (pred) node[black, midway, xshift=6mm, yshift=0mm] {$w_2$};
        \draw[arrow] (pred) -- (loss) {};
        \end{tikzpicture}
        \caption{Observation}
        \label{fig:janzing_observation}
        \end{subfigure}
        \hfill
        \begin{subfigure}{0.24\textwidth}
        \centering
        \begin{tikzpicture}
        \node (pred) at (-1, 0) [variable node] {$\hat{Y}\ustime$};
        \node (loss) at (1, 0) [variable node] {$\ell\ustime$};
        \node (X1) at (-1, 2) [observed node] {$\tilde{x}\ustime\varsub{1}$};
        \node (X2) at (1, 2) [variable node] {$\tilde{X}\ustime\varsub{2}$};
        \node (Z1) at (-1, 4) [latent node] {$X\ustime\varsub{1}$};
        \node (Z2) at (1, 4) [latent node] {$X\ustime\varsub{2}$};
        \node (Z0) at (0, 6) [latent node] {$Z\ustime$};
        \draw[arrow] (Z0) -- (Z1) {};
        \draw[arrow] (Z0) -- (Z2) {};
        \draw[arrow] (Z2) -- (X2) {};
        \draw[arrow] (X1) -- (pred) node[black, midway, xshift=-5mm, yshift=0mm] {$w_1$};
        \draw[arrow] (X2) -- (pred) node[black, midway, xshift=6mm, yshift=0mm] {$w_2$};
        \draw[arrow] (pred) -- (loss) {};
    \end{tikzpicture}
    \caption{Intervention}
    \label{fig:janzing_intervention}
    \end{subfigure}
\caption{The causal structure within a machine learning model. Observed and latent variables are coloured in black and white, respectively.}
\label{fig:janzing_example}
\end{figure}

\paragraph{Interpreting Rewards} We now examine the differences in rewards obtained via the two lifts. 

\begin{theorem} \label{the:causal_effects}
    Marginal contributions derived using the observational conditional expectation as defined in (\ref{eq:observational_shapley}) can be decomposed into both indirect and direct causal effects.
\end{theorem}

\begin{proof}
First, if we let $\Theta$ be the set of all possible permutation of indices in $\Omega$, we can re-formulate the Shapley value in (\ref{eq:shapley_ml}) for feature $i$ at time $t$ as follows:
\begin{equation*}
    \phi_{i}\us{(t)} = \frac{1}{D!} \sum_{\theta \in \Theta} \delta_{i}\us{(t)} (\theta),
\end{equation*}
where now $\delta_{i}\us{(t)} (\theta) = \xi\us{(t)}\varsub{\{j : j \prec_\theta i\}} -  \xi\us{(t)}\varsub{\{j : j \preceq_\theta i\}}$, with $j \prec_\theta i$ meaning $j$ precedes $i$ in permutation $\theta$. Then, using the formulation in (\ref{eq:observational_shapley}), the marginal contribution of feature $i$ for a single permutation $\theta \in \Theta$ derived using the observational lift can be written as
\clearpage
\begin{align*}
    \delta\us{(t), \text{obs}}\varsub{} (\theta) =& \ \xi\us{(t), \text{obs}}\varsub{\ubar{\omega}} - \xi\us{(t), \text{obs}}\varsub{\ubar{\omega} \cup i},\\
    =& \underbracket{\ \int  \ell (\vx\varsub{\ubar{\omega} }\us{(t)}, \vx\varsub{\bar{\omega}\cup i }\us{(t)}) p( \vx\varsub{\bar{\omega}\cup i}\us{(t)} \vert \vx\varsub{\ubar{\omega}}\us{(t)}) d\vx\varsub{\bar{\omega}\cup i}\us{(t)} -\int  \ell (\vx\varsub{\ubar{\omega} \cup i }\us{(t)}, \vx\varsub{\bar{\omega}}\us{(t)}) p(\vx\varsub{\bar{\omega}}\us{(t)} \vert \vx\varsub{\ubar{\omega} \cup i}\us{(t)}) d\vx\varsub{\bar{\omega}}\us{(t)}}\varsub{\textnormal{Total effect}} \\
    =&\underbracket{ \ \int  \ell (\vx\varsub{\ubar{\omega} }\us{(t)}, \vx\varsub{\bar{\omega}\cup i }\us{(t)}) p( \vx\varsub{\bar{\omega}\cup i}\us{(t)} \vert \vx\varsub{\ubar{\omega}}\us{(t)}) d\vx\varsub{\bar{\omega}\cup i}\us{(t)} -\int  \ell (\vx\varsub{\ubar{\omega} \cup i }\us{(t)}, \vx\varsub{\bar{\omega}}\us{(t)}) p(\vx\varsub{\bar{\omega}}\us{(t)} \vert \vx\varsub{\ubar{\omega}}\us{(t)}) d\vx\varsub{\bar{\omega}}\us{(t)}}\varsub{\textnormal{Direct effect}} \\
    &
    +\underbracket{\ \int  \ell (\vx\varsub{\ubar{\omega} \cup i }\us{(t)}, \vx\varsub{\bar{\omega}}\us{(t)}) p(\vx\varsub{\bar{\omega}}\us{(t)} \vert \vx\varsub{\ubar{\omega}}\us{(t)}) d\vx\varsub{\bar{\omega}}\us{(t)} -\int  \ell (\vx\varsub{\ubar{\omega} \cup i }\us{(t)}, \vx\varsub{\bar{\omega}}\us{(t)}) p(\vx\varsub{\bar{\omega}}\us{(t)} \vert \vx\varsub{\ubar{\omega} \cup i}\us{(t)}) d\vx\varsub{\bar{\omega}}\us{(t)}}\varsub{\textnormal{Indirect effect}},
\end{align*}
where $\ubar{\omega} = \{j : j \prec_\theta i\}$ and $\bar{\omega}= \{j : j \succ_\theta i\}$. Thus, the marginal contribution captures two distinct effects. The first is the direct effect on the expected score when feature $i$ is observed and added to the coalition, keeping the distribution of the out-of-coalition features unchanged, in other words, using $p(\vx\varsub{\bar{\omega}}\us{(t)} \vert \vx\varsub{\ubar{\omega}}\us{(t)})$ instead of $p(\vx\varsub{\bar{\omega}}\us{(t)} \vert \vx\varsub{\ubar{\omega}\cup i}\us{(t)})$. The other is the indirect effect on the expected score when the distribution of the out-of-coalition features does change as a result of observing feature $i$, that is, when $p(\vx\varsub{\bar{\omega}}\us{(t)} \vert \vx\varsub{\ubar{\omega}}\us{(t)})$ changes to $p(\vx\varsub{\bar{\omega}}\us{(t)} \vert \vx\varsub{\ubar{\omega}\cup i}\us{(t)})$.
\end{proof}

Following Theorem~\ref{the:causal_effects}, we can see that by replacing conditioning by observation with the marginal
distribution as in (\ref{eq:interventional_shapley}), the indirect effect disappears entirely.
The interventional lift is thus more effective at crediting features upon which the model has an explicit algebraic dependence.
In contrast, the observational lift attributes features equally amongst indirect effects. Some argue that this is illogical as features not explicitly used by the model have the possibility of receiving non-zero allocation \citep{kumar2020problems}.
For instance, in Figure~\ref{fig:janzing_example}, if $w_2$ happens to be $0$, such that $X\varsub{2}\ustime$ has no direct effect on the score, the interventional lift would allcoate this feature no reward. However, it would receive positive reward with the observational lift given the backdoor path via the confounding variable if $w_1 > 0$.

In the context of analytics markets, we know that the predictive performance of the model out-of-sample is contingent upon the availability of features that were used during training, which, in practice, requires data of the support agents to be streamed continuously in a timely fashion.
If a feature was missing, the efficacy of the forecast may drop, the extent to which would relate not to any root causes or indirect effects regarding the data generating process, but rather the magnitude of direct effects on the particular model's output. 
With the interventional lift, comparatively larger rewards would be made to support agents with features to which the predictive performance of the model is most sensitive, providing incentives to avoid data being unavailable, somewhat resembling reserve payments in energy markets, where assets are remunerated for being available in times of need. 
With the observational lift, it would instead be unclear as to whether comparatively larger rewards are consequential of features having an impact on predictive performance, or merely a result of indirect effects through those that do. The interventional lift therefore better aligns with desirable intentions of the market. 

Lastly, these lifts also differ significantly in their computational expense. In particular, computing the observational conditional expectation of $\ell$ is generally intractable, requiring complex and expensive approximations \citep{covert2021explaining}. By contrast, intervening on features can be done relatively simply and efficiently \citep{lundberg2017unified}.

\paragraph{\textbf{Limitations}}
There is, of course, no free lunch, as if features are strongly correlated, conditioning by intervention can lead to model evaluation on points outwith the true data manifold.
This can visualized with the simple illustration in Figure~\ref{fig:manifold}. Whilst intervening on independent features always yields samples within the original manifold, if features are very correlated, there is a possibility of extrapolating beyond the training distribution, where model behavior is unknown. In the remainder of this section we consider what impact this may have on the market outcomes.
Multicollinearity inflates the variance of the coefficients, which can distort the estimated mean when the number of in-sample observations is limited.

The posterior variance of the $i$-th coefficient can be written as $\sigma^2(w_i) = \kappa_i / \xi \vert \mathcal{D}_{t} \vert$, where $\xi$ is the intrinsic noise precision of the target and $\kappa_i$ is the variance inflation factor, given by 
\begin{align*}
    \kappa_i = \textbf{e}_i^\top ( \sum_{t^\prime \leq t} \left(\vx\us{(t)}\right)^\top \vx\us{(t)} )^{-1} \textbf{e}_i, \quad \forall i \in \mathcal{I},
\end{align*}
where $\textbf{e}_i$ is the $i$-th basis vector. 
Whilst $\kappa_i \geq 1$, it has no upper bound, meaning $\kappa_i \mapsto \infty, \, \forall i$, with increasing extent of collinearity.

\begin{figure}
\centering
\includegraphics[width=\linewidth]{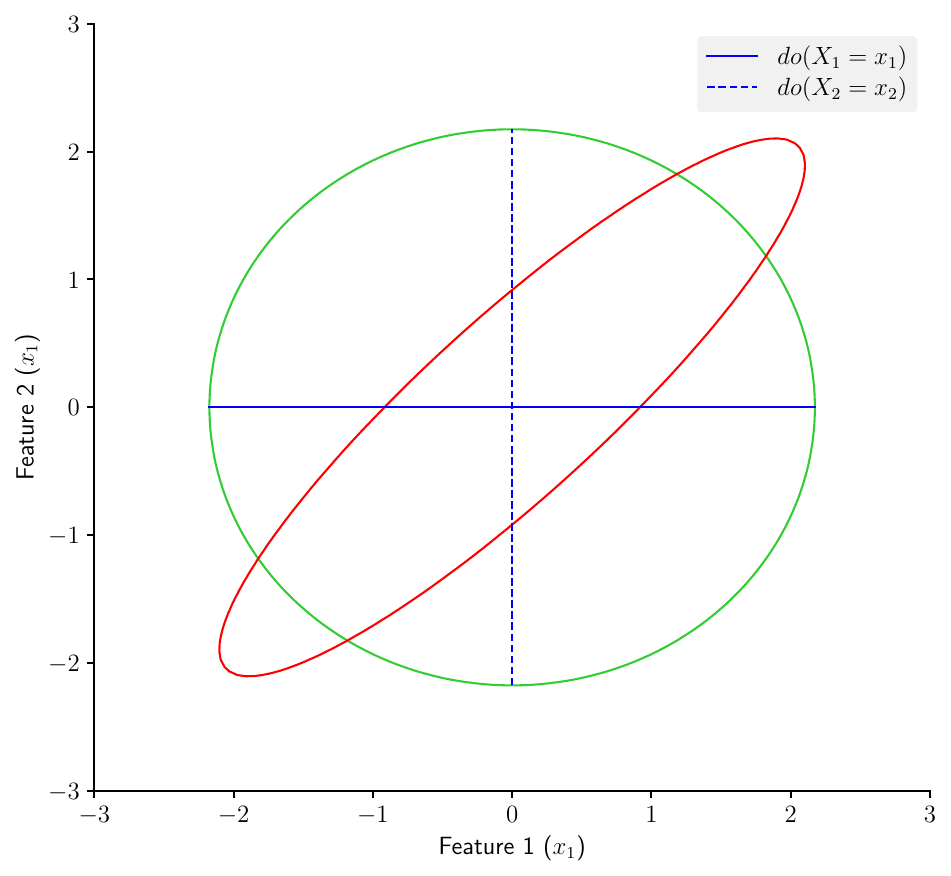}
\caption{Interventions producing points outwith the data manifold.
Green and red lines are level sets within which 0.99 quantile of the training data when features are independent and correlated, respectively. The blue lines represent the data extrapolated as a result of intervening on $X_1$ and $X_2$.}
\label{fig:manifold}
\end{figure}

From a variance decomposition perspective, the Shapley value of feature $i$ equals the variance in the target signal that it explains, such that, $\mathbb{E}[\phi_{i}]\us{(t)} = \left(\mathbb{E}[w_i]\us{(t)}\right)^2 \, var(X_{i}\us{(t)})$, approximating the behaviour of the interventional Shapley value when features are correlated \citep{owen2017shapley}.  
With a Gaussian posterior, the Shapley values follow a noncentral Chi-squared distribution with one degree of freedom. We can write the probability density function for the distribution of the Shapley value for feature $i$ in closed-form as 
\begin{align*}
    &p(\phi_{i}\us{(t)}) \\
    &= var(X_{i}\us{(t)}) var(w_i)\us{(t)} \sum_{n=0}^\infty \frac{e^{\eta/2}}{n!} \left(\frac{\eta}{2}\right)\us{\intN} \chi^2 (1 + 2n), 
\end{align*}
where $var(\cdot)\us{(t)}$ is the estimated variance at time $t$ and the noncentral Chi-squared distribution is seen to simply be given by a Poisson-weighted mixture of central Chi-squared distributions, $\chi^2(\cdot)$, with noncentrality $\eta = \left(\mathbb{E}[w_i]\us{(t)}\right)^2 / var(w_i)\us{(t)}$, for which the  moment generating function is known in closed form.
For feature $i$, the centered second moment is
\begin{align*}
    &var(\phi_i)\us{(t)} = 2var(w_i)\us{(t)} \\
    & \quad \times \left(2 \mathbb{E}[w_i]_t^2 + var(w_i)\us{(t)} \right) \left( var(X_{i}\us{(t)}) \right)^2
\end{align*}
so the variance of the allocation for any feature is a quadratic function of the variance of the corresponding coefficient, thus the variance inflation induced by multicollinearity. That being said, this is only a problem for small sample sizes and vanishes with increasing $t$, as $var(w_i)\us{(t)} \mapsto 0, \, \forall i$ \citep{qazaz1997upper}.
If only a limited number of observations are available, distorted revenues could be remedied using \textit{zero-Shapley} or \textit{absolute-Shapley} proposed in \citet{liu2020absolute}, or restricting evaluations to the data manifold \citep{taufiq2023manifold}. We leave an investigation into these remedies in relation to analytics markets to future work.

\section{Replication Robustness} \label{sec:robust}
In this section, we show that the use of observational conditional expectations in existing works on analytics markets (e.g., \citealp{agarwal2019marketplace}) explains the observed incentives for support agents to submit replicates of their features under different identities. We discuss the downsides of this and prove it can be remedied with the interventional lift

\begin{definition}[Replicate] \label{def:replicate} A replicate of feature $i$ is the original data obfuscated with noise, $x_{i}\ustime + \eta_{i}\ustime$, where $\eta_{i}\ustime$ is drawn from a centered distribution with finite variance, conditionally independent of the target given the feature. 
\end{definition}

We note that, obfuscating a feature in this manner is equivalent to regularizing it’s coefficients during training \citep{bishop1995training},
inducing an endogeneity bias that diminishes the feature’s contribution and, consequently, the reward obtained by
the support agent. This idea underpins the proof of the truthfulness property described in Definition~\ref{def:properties} as provided in \citet{falconer2024bayesian}. However, this property does not
account for the fact that agents could, in theory, submit multiple replicates along with their original feature, each under
a false identity. Whilst this would not impact predictive performance, it allows malicious support agents to increase their own reward and diminish that of others whilst providing no additional improvements if the observational lift is used to compute the Shapley values.

To see this, consider the graphical model in Figure~\ref{fig:replication_data}, with two features $X\varsub{1}\ustime$ and $X\varsub{2}\ustime$, each owned by a separate support agent, $a_1$ and $a_2$, respectively. Suppose these two features are correlated via a latent confounder $Z\ustime$, so that in the model they have equal affects on $\ell\ustime$, with $w_1 = w_2$ in Figure~\ref{fig:replication_model}, following the framework of \citet{janzing2020feature}.
If $a_2$ replicates their feature $K$ times, let $X\ustime\varsub{2\varsub{k}} = X\ustime\varsub{2} + \eta\ustime\varsub{2\varsub{k}}$ be the $k$-th replicate of $X\ustime\varsub{2}$. It is easy to show that each replicate $X\ustime\varsub{2\varsub{k}}$ has no direct causal effect on $\ell$ in the model. 
Specifically, let $\overline{\vx}\ustime = (x\varsub{1}\ustime, x\varsub{2}\ustime)\us{\top} \oplus (x\varsub{2\varsub{1}}\ustime, \dots, x\varsub{2\varsub{K}}\ustime)\us{\top}$ be the complete feature vector including all of the additional replicates, then we can obtain the posterior mean via \textit{maximum a posteriori} estimation, such that after $\intT$ observations:
\begin{align*}
    \frac{\partial}{\partial w\varsub{j}} \log p(\vw \vert y\us{(1)}, \dots, y\us{(\intT)}, &\overline{\vx}\us{(1)}, \dots, \overline{\vx}\us{(\intT)}) \\
    &= \frac{\partial}{\partial w\varsub{j}} \left[- \frac{\beta}{2} \sum\varsub{t=1}\us{\intT} \varepsilon\ustime \right] \\
    &= \beta \sum\varsub{t=1}\us{\intT} \varepsilon\ustime x\varsub{j}\ustime = 0,
\end{align*} 
where $\varepsilon\ustime$ is the residual at time $t$, which can be written as
\begin{align*}
    \varepsilon\ustime
    &= w\varsub{1} x\varsub{1}\ustime + w\varsub{1} x\varsub{2}\ustime + \sum\varsub{k=1}\us{K} w\varsub{2\varsub{k}} x\varsub{2\varsub{k}}\ustime - y\ustime \\
    &= w\varsub{1} x\varsub{1}\ustime + w\varsub{1} x\varsub{2}\ustime + \sum\varsub{k=1}\us{K} w\varsub{2\varsub{k}} \left( x\varsub{2}\ustime + \eta\ustime\varsub{2\varsub{k}} \right) - y\ustime\\
    &= w\varsub{1} x\varsub{1}\ustime + \left(  w\varsub{2} + \sum\varsub{k=1}\us{K} w\varsub{2\varsub{k}} \right) x\varsub{2}\ustime + \sum\varsub{k=1}\us{K} w\varsub{2\varsub{k}} \eta\ustime\varsub{2\varsub{k}} - y\ustime,
\end{align*}
with $w\varsub{2\varsub{k}}$ being the weight associated with the $k$-th replicate of $X\varsub{2}\ustime$ in the latent variable model.

Thus, for every replicate $k \in [K]$, we can write
\begin{align*}
     &\sum\varsub{t=1}\us{\intT} \varepsilon\ustime \eta\varsub{2\varsub{k}}\ustime = w\varsub{1} \sum\varsub{t=1}\us{\intT} x\varsub{1}\ustime \eta\varsub{2\varsub{k}}\ustime + \bigg(w\varsub{2} \, + \sum\varsub{l=1}\us{K} w\varsub{2\varsub{l}} \bigg) \sum\varsub{t=1}\us{\intT} x\varsub{2}\ustime \eta\varsub{2\varsub{k}}\ustime + \sum\varsub{l=1}\us{K} w\varsub{2\varsub{l}} \sum\varsub{t=1}\us{\intT} \eta\ustime\varsub{2{\varsub{l}}} \eta\varsub{2\varsub{k}}\ustime - \sum\varsub{t=1}\us{\intT} y\ustime \eta\varsub{2\varsub{k}}\ustime,
\end{align*}
which equals zero since $\sum\varsub{t=1}\us{\intT} \varepsilon\ustime x\varsub{2}\ustime = 0$ for $w\varsub{2}$.
As $\eta\ustime\varsub{2{\varsub{k}}}$ is white noise and so uncorrelated with the features, target and every other $\eta\ustime\varsub{2{\varsub{l}}}$, where $l \neq k$, it holds in expectation that
\begin{align*}
     \sum\varsub{t=1}\us{\intT} \varepsilon\ustime \eta\varsub{2\varsub{k}}\ustime = w\varsub{2, k} \sum\varsub{t=1}\us{\intT} \eta\varsub{2\varsub{k}}\ustime \eta\varsub{2\varsub{k}}\ustime > 0, 
\end{align*}
so we must have $w\varsub{2\varsub{k}} = 0$ for every $k$ since the noise variance is positive. Thus, each replicate has no \textit{direct} causal effect on the model output.

\begin{figure}
\centering
\makebox[\textwidth][c]{%
    \begin{subfigure}{0.45\textwidth}
    \centering
    \begin{tikzpicture}
        \node (target) at (0, 0) [variable node] {$Y\ustime$};
        \node (X1) at (-2.4, 2) [variable node] {$X\ustime\varsub{1}$};
        \node (X2) at (-0.8, 2) [variable node] {$X\ustime\varsub{2}$};
        \node (X2r1) at (0.8, 2) [variable node] {$X\ustime\varsub{2\varsub{1}}$};
        \node (X2rK) at (2.4, 2) [variable node] {$X\ustime\varsub{2\varsub{K}}$};
        \node (dots) at (1.6, 2) [hidden node] {\footnotesize $\, \cdots$};
        \node (Z0) at (-1.6, 4) [latent node] {$Z\ustime$};
        \draw[arrow] (X1) -- (target) {};
        \draw[arrow] (X2) -- (target) {};
        \draw[arrow] (Z0) -- (X1) {};
        \draw[arrow] (Z0) -- (X2) {};
        \draw[arrow, dashed] (X2r1) -- (target) {};
        \draw[arrow, dashed] (X2rK) -- (target) {};
        \draw[arrow] (X2) -- (X2r1) {};
        \draw[arrow] (X2) to[out=45,in=135] (X2rK) {};
    \end{tikzpicture}
    \caption{Data}
    \label{fig:replication_data}
    \end{subfigure}
    \hfill
    \begin{subfigure}{0.45\textwidth}
    \centering
    \begin{tikzpicture}
        \node (loss) at (0, 0) [variable node] {$\ell\ustime$};
        \node (X1) at (-2.4, 2) [variable node] {$\tilde{X}\ustime\varsub{1}$};
        \node (X2) at (-0.8, 2) [variable node] {$\tilde{X}\ustime\varsub{2}$};
        \node (X2r1) at (0.8, 2) [variable node] {$\tilde{X}\ustime\varsub{2\varsub{1}}$};
        \node (X2rK) at (2.4, 2) [variable node] {$\tilde{X}\ustime\varsub{2\varsub{K}}$};
        \node (dots) at (1.6, 2) [hidden node] {\footnotesize $\, \cdots$};
        \node (Z0) at (-1.6, 6) [latent node] {$Z\ustime$};
        \node (Z1) at (-2.4, 4) [latent node] {$X\ustime\varsub{1}$};
        \node (Z2) at (-0.8, 4) [latent node] {$X\ustime\varsub{2}$};
        \node (Z2r1) at (0.8, 4) [latent node] {$X\ustime\varsub{2\varsub{1}}$};
        \node (Z2rK) at (2.4, 4) [latent node] {$X\ustime\varsub{2\varsub{K}}$};
        \node (dots) at (1.6, 4) [hidden node] {\footnotesize $\, \cdots$};
        \draw[arrow] (X1) -- (loss) node[black, midway, xshift=-5mm, yshift=0mm] {$w_1$};
        \draw[arrow] (X2) -- (loss) node[black, midway, xshift=4mm, yshift=0mm] {$w_2$};
        \draw[arrow, dashed] (X2r1) -- (loss) {};
        \draw[arrow, dashed] (X2rK) -- (loss) {};
        \draw[arrow] (Z0) -- (Z1) {};
        \draw[arrow] (Z0) -- (Z2) {};
        \draw[arrow] (Z2) -- (Z2r1) {};
        \draw[arrow] (Z2) to[out=45,in=135] (Z2rK) {};
        \draw[arrow] (Z1) -- (X1) {};
        \draw[arrow] (Z2) -- (X2) {};
        \draw[arrow] (Z2r1) -- (X2r1) {};
        \draw[arrow] (Z2rK) -- (X2rK) {};
    \end{tikzpicture}
    \caption{Model}
    \label{fig:replication_model}
    \end{subfigure}%
}
\vskip 1em
\makebox[\textwidth][c]{%
    \begin{subfigure}{0.45\textwidth}
    \centering
    \begin{tikzpicture}
        \node (loss) at (0, 0) [variable node] {$\ell\ustime$};
        \node (X1) at (-2.4, 2) [variable node] {$\tilde{X}\ustime\varsub{1}$};
        \node (X2) at (-0.8, 2) [variable node] {$\tilde{X}\ustime\varsub{2}$};
        \node (X2r1) at (0.8, 2) [observed node] {$\tilde{x}\ustime\varsub{2\varsub{1}}$};
        \node (X2rK) at (2.4, 2) [observed node] {$\tilde{x}\ustime\varsub{2\varsub{K}}$};
        \node (dots) at (1.6, 2) [hidden node] {\footnotesize $\, \cdots$};
        \node (Z0) at (-1.6, 6) [latent node] {$Z\ustime$};
        \node (Z1) at (-2.4, 4) [latent node] {$X\ustime\varsub{1}$};
        \node (Z2) at (-0.8, 4) [latent node] {$X\ustime\varsub{2}$};
        \node (Z2r1) at (0.8, 4) [latent node] {$X\ustime\varsub{2\varsub{1}}$};
        \node (Z2rK) at (2.4, 4) [latent node] {$X\ustime\varsub{2\varsub{K}}$};
        \node (dots) at (1.6, 4) [hidden node] {\footnotesize $\, \cdots$};
        \draw[arrow] (X1) -- (loss) node[black, midway, xshift=-5mm, yshift=0mm] {$w_1$};
        \draw[arrow] (X2) -- (loss) node[black, midway, xshift=4mm, yshift=0mm] {$w_2$};
        \draw[arrow, dashed] (X2r1) -- (loss) {};
        \draw[arrow, dashed] (X2rK) -- (loss) {};
        \draw[arrow] (Z0) -- (Z1) {};
        \draw[arrow] (Z0) -- (Z2) {};
        \draw[arrow] (Z2) -- (Z2r1) {};
        \draw[arrow] (Z2) to[out=45,in=135] (Z2rK) {};
        \draw[arrow] (Z1) -- (X1) {};
        \draw[arrow] (Z2) -- (X2) {};
        \draw[arrow] (Z2r1) -- (X2r1) {};
        \draw[arrow] (Z2rK) -- (X2rK) {};
    \end{tikzpicture}
    \caption{Observation}
    \label{fig:replication_observation}
    \end{subfigure}
    \hfill
    \begin{subfigure}{0.45\textwidth}
    \centering
        \begin{tikzpicture}
        \node (loss) at (0, 0) [variable node] {$\ell\ustime$};
        \node (X1) at (-2.4, 2) [variable node] {$\tilde{X}\ustime\varsub{1}$};
        \node (X2) at (-0.8, 2) [variable node] {$\tilde{X}\ustime\varsub{2}$};
        \node (X2r1) at (0.8, 2) [observed node] {$\tilde{x}\ustime\varsub{2\varsub{1}}$};
        \node (X2rK) at (2.4, 2) [observed node] {$\tilde{x}\ustime\varsub{2\varsub{K}}$};
        \node (dots) at (1.6, 2) [hidden node] {\footnotesize $\, \cdots$};
        \node (Z0) at (-1.6, 6) [latent node] {$Z\ustime$};
        \node (Z1) at (-2.4, 4) [latent node] {$X\ustime\varsub{1}$};
        \node (Z2) at (-0.8, 4) [latent node] {$X\ustime\varsub{2}$};
        \node (Z2r1) at (0.8, 4) [latent node] {$X\ustime\varsub{2\varsub{1}}$};
        \node (Z2rK) at (2.4, 4) [latent node] {$X\ustime\varsub{2\varsub{K}}$};
        \node (dots) at (1.6, 4) [hidden node] {\footnotesize $\, \cdots$};
        \draw[arrow] (X1) -- (loss) node[black, midway, xshift=-5mm, yshift=0mm] {$w_1$};
        \draw[arrow] (X2) -- (loss) node[black, midway, xshift=4mm, yshift=0mm] {$w_2$};
        \draw[arrow] (Z0) -- (Z1) {};
        \draw[arrow] (Z0) -- (Z2) {};
        \draw[arrow] (Z2) -- (Z2r1) {};
        \draw[arrow] (Z2) to[out=45,in=135] (Z2rK) {};
        \draw[arrow] (Z1) -- (X1) {};
        \draw[arrow] (Z2) -- (X2) {};
    \end{tikzpicture}
    \caption{Intervention}
    \label{fig:replication_intervention}
    \end{subfigure}%
}
\caption{Effects induced by replicating $X\varsub{2}\ustime$. For brevity we have omitted the direct path $\hat{Y}\ustime \to \ell\ustime$.}
\label{fig:replication}
\end{figure}

Despite this, replicates will still be rewarded by the observational lift due to indirect causal effects, which give rise to correlations between them and the target. Specifically, without any replication, (i.e., $K=0$), since $w_1 = w_2$, the reward to each support agent will be $r\varsub{c}\ustime/2$, where recall $r\varsub{c}\ustime$ is the market revenue. For $K > 0$, with the same logic the reward allocated to each feature is $r\varsub{c}\ustime/(2 + K)$ since even though no direct effects exist, contributions are split equally via the backdoor paths. Thus the resulting rewards are  
\begin{align*}
    r\varsub{a\varsub{1}}\ustime = \frac{r\varsub{c}\ustime}{2 + K} \quad \text{and} \quad r\varsub{a\varsub{2}}\ustime = \sum\varsub{k=1}\us{1 + K} \frac{r\varsub{c}\ustime}{2 + K} = \frac{r\varsub{c}\ustime (1 + K)}{2 + K},
\end{align*}
hence $a_2$ can maliciously replicate their data many times and increase their overall reward, whilst diminishing that of $a_1$, since $ r\varsub{a\varsub{1}}\ustime \to 0$ as $K \to \infty$.

Let $\overline{\vx}\ustime \in \mathbb{R}^{M + D + \sum\varsub{i\in\Omega} K_i}$ be the augmented feature vector, with an extended index set $\overline{\Omega}$, after support agent $a \in \mathcal{A}_{-c}$ replicates feature $i$ a total of $K_i$ times. 
In \citet{agarwal2019marketplace}, the Shapley value is modified to penalize similar features. Specifically, they propose \textit{Robust-Shapley}:
\begin{align*}
    \phi_{i}\us{(t), \text{robust}} = \phi_{i}\ustime \exp \left(-\gamma \sum_{j \in \Omega} \texttt{sim}\left({X_{i}\ustime, X_{j}\ustime} \right) \right),
\end{align*}
where $\texttt{sim}(\cdot,\cdot)$ is some measure of similarity (e.g., cosine similarity). This penalizes similar features so as to remove any incentive for replication, satisfying the following definition of replication robustness.

\begin{definition}[\textbf{Weakly Replication-robust}] \label{def:weak_robustness} The regression market is \textbf{weakly} robust to replication if $\overline{r}\ustime\varsub{a} \leq r\ustime\varsub{a}$, where $\overline{r}\ustime\varsub{a}$ is the reward of $a \in \mathcal{A}\varsub{-c}$ as described in \eqref{eq:reward_rule} after using $\overline{\vx}\ustime$ instead.
\end{definition}

This implies that agents who submit replicates should obtain weakly less reward than before. 
However, the penalty applies not only to replicated features but also to those that are naturally correlated. This results in a loss of budget balance, the extent of which depends on the similarity metric and the value of $\gamma$. Moreover, Definition~\ref{def:weak_robustness} remains vulnerable to spiteful agents, which is why we refer to this definition as \textit{weakly} robust. A similar result appears in \citet{han2022replication}, who show that using the Banzhaf value \citep{lehrer1988axiomatization} instead inherently leads to weak replication-robustness. 

We now introduce the following stronger definition of robustness to replication.

\begin{definition}[\textbf{Strictly Replication-robust}] \label{def:strict_robustness} The regression market is \textbf{strictly} robust to replication if $\overline{r}\ustime\varsub{a} = r\ustime\varsub{a}$.
\end{definition}

If an allocation policy is \textit{strictly} replication-robust, rewards remain unchanged when features are replicated, removing any incentive to do so and providing protection against spiteful agents. 

\begin{proposition}
    With the proposed regression framework and Shapley value-based revenue allocation, regression markets using the interventional lift are strictly replication-robust.
\end{proposition}
\begin{proof}
    With Definition~\ref{def:replicate}, each replicate in $\bar{\vx}\us{(t)}$ only induces an indirect effect on the target. However, from Theorem~\ref{the:causal_effects}, we know that the interventional lift only captures direct effects. Therefore, for each of the replicates, we write the marginal contribution for a single permutation $\theta \in \Theta$ as
    \begin{align*}
        \delta\varsub{i}\us{(t), \textrm{int}} (\theta) &= \xi\us{(t), \textrm{int}}\varsub{\ubar{\omega}} - \xi\us{(t), \textrm{int}}\varsub{\ubar{\omega} \cup i}, \\
        &\int  \ell ( \vx\varsub{\ubar{\omega} }\us{(t)}, \vx\varsub{\bar{\omega} \cup i }\us{(t)}) p( \vx\varsub{\bar{\omega} \cup i}\us{(t)} \vert \vx\varsub{\ubar{\omega}}\us{(t)}) d\vx\varsub{\bar{\omega} \cup i}\us{(t)} -\int  \ell ( \vx\varsub{\ubar{\omega} \cup i }\us{(t)}, \vx\varsub{\bar{\omega} }\us{(t)}) p(\vx\varsub{\bar{\omega}}\us{(t)} \vert \vx\varsub{\ubar{\omega}}\us{(t)}) d\vx\varsub{\bar{\omega}}\us{(t)}, \\
        &= 0, \quad \forall i \in \overline{\Omega} \setminus \Omega,
    \end{align*}
    and therefore $\phi_{i} \propto  \sum_{\theta \in \Theta} \Delta_i (\theta) = 0$ for each of the replicates. For the original features, any direct effects will remain unchanged. This leads to
    \begin{align*}
        \overline{r}\varsub{a}\us{(t)} &= \sum_{i \in \Omega} \lambda \mathbb{E} [ \phi_i ]\us{(t)} + \sum_{j \in \overline{\Omega} \setminus \Omega} \lambda \underbracket{\mathbb{E} [ \phi_j ]\us{(t)}}_{=0} = r_{a}\us{(t)}, \ \forall a \in \mathcal{A}_{-c},
    \end{align*}
    showing that by replacing the conventional observational lift with the interventional lift, the Shapley value-based allocation is robust to replication \textit{and} spitefulness by design, by removing the backdoor paths as illustrated in Figure~\ref{fig:replication_intervention} compared with Figure~\ref{fig:replication_observation}.
\end{proof}

\section{Experimental Analysis} \label{sec:experimental_analysis}
We now validate our findings on a real-world case study.\footnote{Our code has been made available at: \url{https://github.com/tdfalc/regression-markets}}
We use an open source dataset to facilitate reproduction of our work, namely the Wind Integration National Dataset (WIND)
Toolkit, detailed in \citet{draxl2015wind}.
Our setup is a stylised continuous electricity market where agents---in our case, wind producers---need to notify the system operator of their expected electricity generation in a forward stage, one hour ahead of delivery, for which they receive a fixed price per unit. In real-time, they receive a penalty for deviations from the scheduled production, thus their downstream revenue is an explicit function of forecast accuracy.

\paragraph{\textbf{Data Description}}
This dataset contains wind power measurements simulated for 9 wind farms in South Carolina (USA), all located within 150 km of each other---see Table~\ref{tab:freq} for a characteristic overview. Although this data is not exactly \textit{real}, it effectively captures the spatio-temporal aspects of wind power production, with the added benefit of remaining free from any spurious measurements, as can often be the case with real-world datasets. Measurements are available for a period of 7 years, from 2007 to 2013, with an hourly granularity, which we normalize to take values in the range of $[0, 1]$. 

Each wind farm is considered a market agent. For simplicity, we let $a_1$ be the central agent, however in practice each could assume this role in parallel. We assume each agent to have only $1$ feature, namely the $1$-hour lag of their power measurements---for wind power forecasting, the lag not only captures the temporal correlations of the production at a specific site, but also indirectly encompasses the spatial dependencies amongst neighboring sites due to the natural progression of wind. To illustrate this, we plot the location of each site in Figure~\ref{fig:map}.
We see that the measurements at sites directly neighbouring $a_1$ have the largest dependency, which then decreases for the sites further away. 

\begin{figure}
\centering
\includegraphics[width=\textwidth]{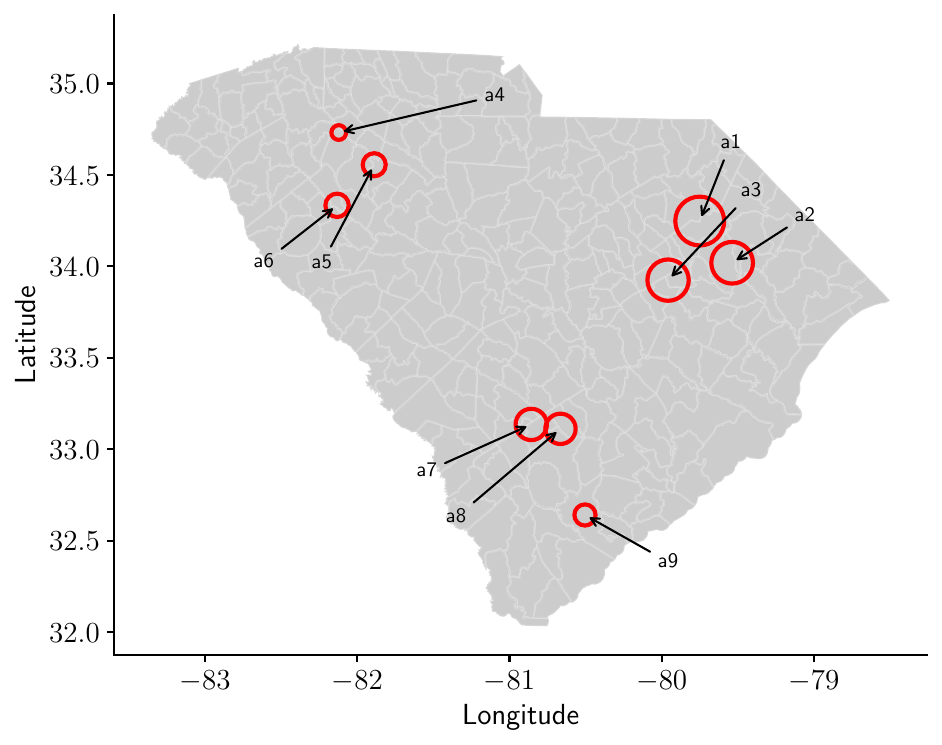}
\caption{Geographic location of each wind farm. The point sizes indicate the relative correlation between the measurements at each site and that of the central agent, $a_1$.}
\label{fig:map}
\end{figure}

\begin{table}[!t]
  \centering
  \caption{Agents and corresponding site characteristics considered in South Carolina (USA). $C_{\textrm{f}}$ denotes the capacity factor and $P$ the nominal capacity. The identify number is that from the WIND Toolkit database.}
  \vspace{5mm}
  \label{tab:freq}
  \begin{tabular}{lrrrr}
    \toprule
    Agent & Id. & $C_{\textrm{f}}$ (\%)  & $P$ (MW) \\
    \midrule
    $a_1$ & 4456 & 34.11 & 1.75 \\
    $a_2$ & 4754 & 35.75 & 2.96 \\
    $a_3$ & 4934 & 36.21 & 3.38 \\
    $a_4$ & 4090 & 26.60 & 16.11 \\
    $a_5$ & 4341 & 28.47 & 37.98 \\
    $a_6$ & 4715 & 27.37 & 30.06 \\
    $a_7$ & 5730 & 34.23 & 2.53 \\
    $a_8$ & 5733 & 34.41 & 2.60 \\
    $a_9$ & 5947 & 34.67 & 1.24 \\
  \bottomrule
\end{tabular}
\end{table}

\paragraph{\textbf{Methodology}}
We use the regression framework described in Section~\ref{sec:preliminaries}, with an \textit{Auto-Regressive with eXogenous input} model, such that each agent is assumed to own a single feature, namely a 1-hour lag of their power measurement.
We are interested in assessing market outcomes rather than competing with state-of-the-art forecasting methods, so we use a very short-term lead time (i.e., 1-hour ahead), permitting fairly simple time-series analyses.
We focus on assessing rewards rather than competing with state-of-the-art forecasting methods, so we use a very short-term lead time, permitting fairly simple time-series analyses.
Nevertheless, our mechanism readily allows more complex models for those aiming to capture specific intricacies of wind power production, for instance the bounded extremities of the power curve \citep{pinson2012very}. 

We perform a pre-screening, such that given the redundancy between the lagged measurements of $a_2$ and $a_3$ with that of $a_1$, we remove them from the market in line with our assumptions. At every time step, once a new observation of the target signal arrives, the previous time step’s forecast is applied for out-of-sample market clearing. Simultaneously, the posterior is updated, the in-sample market is cleared, and a forecast for the next time step is generated.
We clear both markets considering each agent is honest, that is, they each provide a single report of their true data. Next, we re-clear the markets, but this time assuming agent $a_4$ is malicious, replicating their data, thereby submitting multiple separate features to the market to increase their revenue. This problem size doesn't require approximate Shapley values, but recall findings hold either way, and generalize theoretically to arbitrary numbers of agents. 


\begin{figure}[!t]
    \centering
    \begin{subfigure}{\textwidth}
        \centering
        \includegraphics[width=\linewidth]{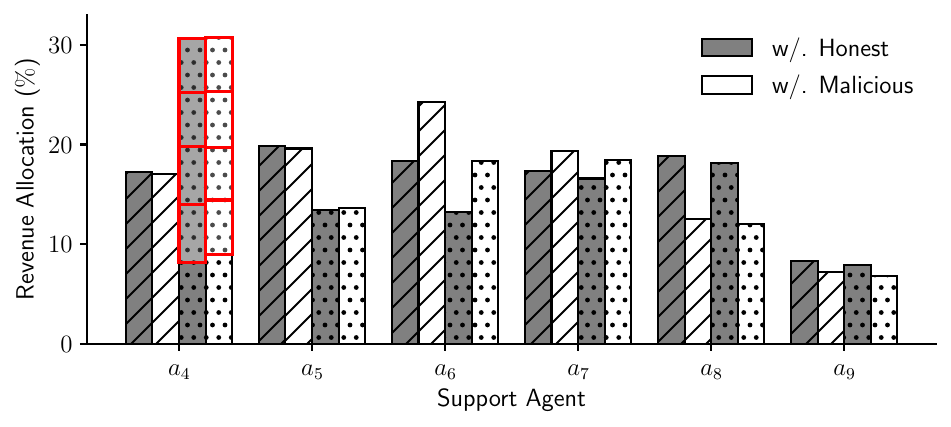}
        \caption{\textit{Observational}: Revenue of $a_4$ increases due to indirect effects induced by the replicates.}
        \label{fig:observational}
    \end{subfigure}
    \hfill
    \begin{subfigure}{\textwidth}
        \centering
        \includegraphics[width=\linewidth]{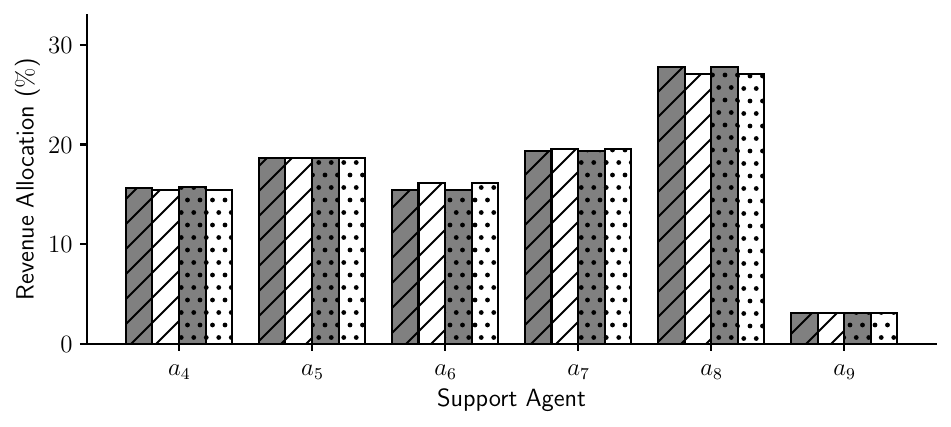}
        \caption{\textit{Interventional}: Revenue of $a_4$ remains the same by accounting only for direct effects.}
        \label{fig:interventional}
    \end{subfigure}
    \caption{Revenue allocations for each support agent. Results for both (a) observational and (b) interventional lifts, when agent $a_4$ is honest ($\slash \slash$) and malicious ($\circ$) by replicating their feature. The blue and green bars correspond to in-sample and out-of-sample market stages, respectively. The revenue split amongst replicates is depicted by the stacked bars highlighted in red.}
    \label{fig:revenue_allocaions}
\end{figure}

\paragraph{\textbf{Results}} We set the central agent's to valuation to $\lambda = 0.5$ USD per time step and per unit improvement in $h$, for both in-sample and out-of-sample market stages. However, we are primarily interested in reward allocation rather than the magnitude---see \citet{pinson2022regression} for a complete analysis of the monetary incentive to each agent participating in the market. Overall the expected in-sample and out-of-sample losses improved by 10.6\% and 13.3\% respectively with the help of the support agents. This improvement is unaffected bu the number of replicates, since they provide no additional information.

Setting $K=4$, in Figure~\ref{fig:revenue_allocaions}, we plot the expected allocation for each agent both with and without the malicious behavior of agent $a_4$, for each lift. When $a_4$ is honest, we observe that the observational lift spreads credit relatively evenly amongst features, suggesting that many of them have similar indirect effects on the target. The interventional lift favours agents $a_7$ and $a_8$, which, as one would expect, own the features with the most spatial correlation with the target. In this market, most of the additional revenue of agent $a_8$ appears to be lost from agent $a_9$ compared with the observational lift, suggesting that whilst these features are correlated, it is agent $a_8$ with the greatest direct effect, which is intuitive given their geographic location.

\begin{figure}
\centering
\includegraphics[width=\textwidth]{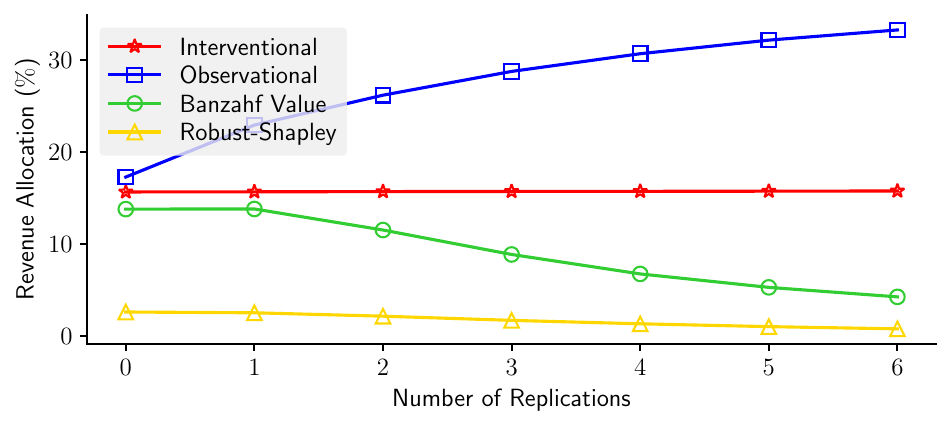}
\caption{Revenue allocation of agent $a_4$ with increasing number of replicates.}
\label{fig:comparison}
\end{figure}

When agent $a_4$ replicates their data, with the observational lift, agents $a_5$ to $a_8$ earn less, whilst agent $a_4$ earns more. This shows that this lift indeed spreads rewards proportionally amongst indirect effects, of which there are four more due to the replicates, and so the malicious agent out-earns the others. Since the interventional lift only attributes direct effects, each replicate gets zero reward, so the malicious agent is no better off than before. Rewards were consistent between in-sample and out-of-sample, likely due to the large sample size and limited nonstationarities within the data.

To compare our work against current literature, in Figure~\ref{fig:comparison} we plot the allocation of agent $a_4$ with increasing number of replicates. Here, \textit{Robust-Shapley} and \textit{Banzahf Value} refer to both the penalization approach of \citet{agarwal2019marketplace} and the use of another semivalue in \citet{han2022replication}, respectively. With the observational lift, the proportion of revenue obtained increases with the number of replicates, as in the previous experiment. With \textit{Robust-Shapley}, the allocation indeed decreases with the number of replicates, demonstrating this approach is \textit{weakly} replication-robust, but is considerably less compared with the other approaches since natural similarities are also penalized. The authors argue this is an incentive for provision of unique information, but this allows agents to be spiteful. The \textit{Banzahf Value} is strictly robust to replication for $K=1$, but only weakly for $K\geq2$. Lastly, unlike these methods, our proposed use of the interventional lift remains strictly replication-robust throughout as expected, with agent $a_4$ not able to benefit from replicating their feature, without penalizing the other agents.

\section{Conclusions} \label{sec:conclusions}
Many machine learning tasks could benefit from using the data owned by others, however convincing firms to share information, even if privacy is assured, poses a considerable challenge. Rather than relying on data altruism, analytics markets are recognized as a promising way of providing incentives for data sharing, many of which use Shapley values to allocate revenue. Nevertheless, there are a number of open challenges that remain before such mechanisms can be used in practice, one of which is vulnerability to strategic replication, which we showed leads to undesirable reward allocation and restricts the practical viability of these markets.

We introduced a general framework for analytics markets for supervised learning problems that subsumes many of these existing proposals. We demonstrated that there are several different ways to formulate a machine learning task as cooperative game and analysed their differences from a causal perspectives. We showed that use of the observational lift to value a coalition is the source of these replication incentives, which many works have tried to remedy through penalization methods, which facilitate only \textit{weak} robustness. Our main contribution is an alternative algorithm for allocating rewards that instead uses interventional conditional probabilities. Our proposal is robust to replication without comprising market properties such as budget balance. This is a step towards making Shapley value-based analytics markets feasible in practice.

From a causal perspective, the interventional lift has additional potential benefits, including reward allocations that better represent the reliance of the model on each feature, providing an incentive for timely and reliable data streams for useful features, that is, those with greater influence on predictive performance. It is also favorable with respect to computational expenditure. 
That said, when it comes to data valuation, the Shapley value is not without its limitations---it is not generally well-defined in a machine learning context and requires strict assumptions, not to mention its computational complexity. This should incite future work into alternative mechanism design frameworks, for example those based on non-cooperative game theory instead.

\section*{Acknowledgements}
We are grateful to Shobhit Singhal (DTU) for valuable discussions and insightful feedback, which contributed to strengthening the rigor of this work.

\bibliographystyle{plainnat}
\bibliography{bib}



\end{document}